\documentclass[journal]{IEEEtran}

\usepackage{amsmath}
\usepackage{amssymb}
\usepackage{amsthm}
\usepackage{bm}
\usepackage{algorithm}
\usepackage{algpseudocode}
\usepackage{cite}
\usepackage{enumerate}
\usepackage{ifthen}
\usepackage[T1]{fontenc}
\usepackage{url}
\usepackage{txfonts}

\DeclareMathOperator{\diag}{diag}
\DeclareMathOperator{\toep}{toep}
\DeclareMathOperator{\vdm}{vdm}

\DeclareMathOperator{\lead}{lead}
\DeclareMathOperator{\rep}{rep}
\DeclareMathOperator{\lat}{\Lambda}

\newcommand{\F}{\mathbb{F}}
\newcommand{\N}{\mathbb{N}}

\newcommand{\Z}{\mathbb{Z}}

\newcommand{\Input}{\item[\textsc{Input:}]}
\newcommand{\Output}{\item[\textsc{Output:}]}
\newcommand{\ZComment}[1]{\(\triangleright\) #1}

\newcommand{\To}{\; \textbf{to} \;}
\newcommand{\SuchThat}{\; \textbf{such that} \;}

\newcommand{\tp}{\textrm{\tiny T}}

\newtheorem{theorem}{Theorem}
\newtheorem{definition}{Definition}

\newboolean{anonymous}
\setboolean{anonymous}{false}

\begin{document}

\title{\bf Decoding square-free Goppa codes over $\F_p$}

\ifthenelse{\boolean{anonymous}}
{ 
\author{
    Anonymized for submission
}
} 
{ 
\author{
    Paulo~S.~L.~M.~Barreto,
    Rafael~Misoczki,
    and
    Richard~Lindner
\thanks{
    P. Barreto is supported by the Brazilian National Council for Scientific and Technological Development (CNPq) under research productivity grant 303163/2009-7.
}
\thanks{
    P. Barreto is with the Department of Computer and Digital Systems Engineering,
    Escola Polit\'{e}cnica, Universidade de S\~{a}o Paulo, Brazil.
    (e-mail: \texttt{pbarreto@larc.usp.br})
}
\thanks{
    R. Misoczki is with SECRET team, INRIA Paris-Rocquencourt, France.
    (e-mail: \texttt{rafael.misoczki@inria.fr})
}
\thanks{
    R. Lindner is with the Department~of~Computer~\mbox{Science},
	  Technische Universit\"{a}t Darmstadt, Darmstadt, Germany.
	  (e-mail: \texttt{rlindner@cdc.informatik.tu-darmstadt.de})
}
}
} 

\markboth{
}{Barreto \MakeLowercase{\textit{et al.}}: Decoding square-free Goppa codes over $\F_p$}

\maketitle

\begin{abstract}
We propose a new, efficient non-deterministic decoding algorithm for square-free Goppa codes over $\F_p$ for any prime $p$. If the code in question has degree $t$ and the average distance to the closest codeword is at least $(4/p)t + 1$, the proposed decoder can uniquely correct up to $(2/p)t$ errors with high probability. The correction capability is higher if the distribution of error magnitudes is not uniform, approaching or reaching $t$ errors when any particular error value occurs much more often than others or exclusively. This makes the method interesting for (semantically secure) cryptosystems based on the decoding problem for permuted and punctured Goppa codes.
\end{abstract}


\section{Introduction}\label{sec:intro}

\IEEEPARstart{P}{ublic-key} cryptosystems based on coding theory, known for nearly as long as the very concept of asymmetric cryptography itself, have recently been attracting renewed interest because of their apparent resistance even against attacks mounted with the help of quantum computers, constituting a family of so-called post-quantum cryptosystems~\cite{bernstein-buchmann-dahmen}. However, not all error-correcting codes are suitable for cryptographic applications. The most commonly used family of codes for such purposed is that of Goppa codes, which remain essentially unharmed by cryptanalysis efforts despite considerable efforts and progress in the area.

Introduced in 1970, Goppa codes~\cite{goppa} are a subfamily of alternant codes, i.e. subfield subcodes of Generalized Reed-Solomon codes.
Let $q = p^m$ for some prime $p$ and some $m > 0$. A Goppa code $\Gamma(L, g)$ over $\F_p$ is determined by a sequence $L \in \F_q^n$ of distinct values, and a polynomial $g \in \F_q[x]$ of degree $t := \deg(g)$ whose roots are disjoint from $L$. Goppa codes have by design a minimal distance at least $t + 1$ by virtue of being alternant. Certain codes are known to have better minimum distances than this lower bound. Thus, binary Goppa codes where $g$ is square-free are known to have a larger minimum distance of at least $2t + 1$ instead. A family of codes where $g$ is not square-free have minimum distance at least $t + \gamma - 1$ for some $2 < \gamma < t - 1$, which is known as the Hartmann-Tzeng bound for Goppa codes~\cite{tzeng-hartmann,hartmann-tzeng}.

The class of Sugiyama-Kasahara-Hirasawa-Namekawa codes~\cite{sugiyama-kasahara-hirasawa-namekawa} where $g = h^{r-1}$ for some square-free monic polynomial $h \in \F_q[x]$ and some power $r$ of $p$ dividing $q$, which constitute a proper superclass of the so-called ``wild'' codes where $h$ is restricted to being irreducible~\cite{bernstein-lange-peters:wild}, have minimum distance at least $r\deg(h) + 1$ rather than $(r - 1)\deg(h) + 1$.
Although it is known that the minimum distance of a Goppa code of degree $t$ is at least $t + 1$ and there are known cases where it is higher (up to $2t + 1$, as it happens for binary square-free Goppa codes), systematically determining the true minimum distance of any given subfamily of Goppa codes remains largely an open problem, yet it is an important metric as it determines not only how many errors can always be uniquely corrected, but indirectly the security level and the key sizes of the cryptosystems based on each given code.

Apart from brute force, known decoding methods for alternant codes can in general correct only about half as many errors as a binary square-free Goppa code is in principle able to correct~\cite{berlekamp,sugiyama-kasahara-hirasawa-namekawa:decoding} (see also~\cite{macwilliams-sloane}). Even the Guruswami-Sudan algorithm~\cite{guruswami-sudan}, which exceeds the $t/2$ limit, can only correct about $n - \sqrt{n(n - t)} \approx t/2 + (t/2)^2/(2n - t)$ errors. In contrast, Patterson's algorithm can correct all $t$ design errors of binary Goppa codes, as can an alternant decoder using the equivalence $\Gamma(L, g) = \Gamma(L, g^2)$ albeit at a larger computational cost. Bernstein's list decoding method~\cite{bernstein} goes somewhat further, attaining a correction capability of $n - \sqrt{n(n - 2t - 2)} \approx t + 1 + (t + 1)^2/2(n - t - 1)$ errors for binary irreducible Goppa codes, although decoding is ambiguous if the actual distance is not proportionally higher. Similar techniques can in principle correct about $n - \sqrt{n(n - rt)} \approx rt/2 + (rt/2)^2/(2n - rt)$ errors for wild codes~\cite{bernstein-lange-peters:wild}. Bernstein's method does not reach the $q$-ary Johnson radius, but a more recent algorithm by Augot \emph{et al.} does so in the binary case \cite{augot-barbier-couvreur}.

\subsection{Our Results}\label{sec:our-contrib}

Our contribution in this paper is a non-deterministic decoding algorithm for square-free Goppa codes over $\F_p$ for any prime $p$. The method generalizes Patterson's approach and can potentially correct up to $(2/p)t$ errors, on the condition that a suitable short vector can be found in a certain polynomial lattice. In particular, our method corrects $(2/3)t$ errors in characteristic~3, exceeding the $t/2$ barrier when the average distance to the closest codeword is at least $(4/3)t + 1$.
In experiments conducted to assess the practical behaviour of our proposal, the result of the decoding is observed to be unique with overwhelming probability for irreducible ternary Goppa codes chosen uniformly at random,
hinting that, for the vast majority of such codes, the average distance to the closest codeword is sufficiently higher than the ensured minimum distance.
Besides, our proposal can probabilistically correct a still larger number of errors that approaches and reaches $t$ depending on the distribution of error magnitudes. For instance, the method corrects up to $t$ errors with high probability if all error magnitudes are known to be equal.

This feature outperforms even Sugiyama-Kasahara-Hirasawa-Namekawa and wild codes and the associated decoding methods, and is particularly interesting for cryptographic applications like McEliece encryption~\cite{mceliece} under the Fujisaki-Okamoto or similar semantic security transform~\cite{fujisaki-okamoto}, where error magnitudes can be chosen by convention to be all equal. In that case, even if an attacker could somehow derive a generic alternant decoder from the public code that is typical in such systems (a strategy exploited e.g. in~\cite{faugere-otmani-perret-tillich}), he will not be able to correct more than about $t/2$ errors out of roughly $t$ that can be corrected with the private trapdoor enabled by our proposal, facing an infeasible workload of about $(p-1)\binom{n}{t/2}/\binom{t}{t/2}$ guesses to mount a complete attack. This makes Goppa codes in odd characteristic, which have already been shown to sport some potential security advantages over binary ones~\cite{peters}, even more attractive in practice.

For the benefit of implementors, we describe a dedicated version of the Mulders-Storjohann algorithm to convert the particular lattice basis encountered during the decoding process to weak Popov form.
The computational complexity of this step is then shown to be $O(p^3 t^2)$.

\subsection{Organization of the Paper}

The remainder of this document is organized as follows. We provide basic notions in Section~\ref{sec:prelim}. We recapitulate Patterson's decoding algorithm for binary irreducible Goppa codes in Section~\ref{sec:patterson}, and extend it to square-free codes in characteristic~$p$ in Section~\ref{sec:extension}, showing that it can correct $(2/p)t$ errors in general and up to $t$ errors depending on the distribution of error magnitudes. We conclude in Section~\ref{sec:conclusion}.

\section{Preliminaries}\label{sec:prelim}

Matrix indices will start from 0 throughout this paper, unless otherwise stated. Let $p$ be a prime and let $q = p^m$ for some $m > 0$. The finite field of $q$ elements is written $\F_q$. For sequences of elements $(g_1, \dots, g_t) \in \F_q^t$, $(L_0, \dots, L_{n-1}) \in \F_q^n$ and $(d_0, \dots, d_{n-1}) \in \F_q^n$ for some $t, n \in \N$, we denote by $\toep(g_1, \dots, g_t)$ the $t \times t$ Toeplitz matrix with elements $T_{ij} := g_{t - i + j}$ for $j \leqslant i$ and $T_{ij} := 0$ otherwise; by $\vdm_t(L_0, \dots, L_{n-1})$ the $t \times n$ Vandermonde matrix with elements $V_{ij} := L_j^i$, $0 \leqslant i < t$, $0 \leqslant j < n$; and by $\diag(d_0, \dots, d_{n-1})$ the diagonal matrix with diagonal elements $D_{jj} := d_j$, $0 \leqslant j < n$.

\subsection{Error Correcting Codes}

Let $L = (L_0, \dots, L_{n-1}) \in \F_q^n$ be a sequence (called the \emph{support}) of $n \leqslant q$ distinct elements, and let $g \in \F_q[x]$ be an irreducible monic polynomial of degree $t$ such that $g(L_i) \neq 0$ for all $i$. For any word $e \in \F_p^n$ we define the corresponding \emph{Goppa syndrome} polynomial $s_e \in \F_q[x]$ to be:
\[
s_e(x) = \sum_{i=0}^{n-1}{\dfrac{e_i}{x - L_i} \mod{g(x)}}.
\]
Thus the syndrome is a linear function of $e$.
The $[n, \geqslant n - mt, \geqslant t+1]$ \emph{Goppa code} over $\F_p$ with support $L$ and generator polynomial $g$ is the kernel of the syndrome function applied to elements from $\F_p$, i.e. the set $\Gamma(L, g) := \{e \in \F_p^n \mid s_e \equiv 0 \mod{g}\}$.

Writing $s_e(x) := \sum_i{s_i x^i}$ for some $s \in \F_q^{t}$, one can show that $s^{\tp} = H e^{\tp}$ where the \emph{parity-check matrix} $H$ has the form 
\begin{equation}\label{eq:parity-check}
\begin{array}{rcl}
H &=& \toep(g_1, \dots, g_t)\\
  &\cdot& \vdm_t(L_0, \dots L_{n-1})\\
  &\cdot& \diag(g(L_0)^{-1}, \dots, g(L_{n-1})^{-1})
\end {array}
\end{equation}
Thus $H = TVD \in \F_q^{t \times n}$, where $T \in \F_q^{t \times t}$ is a Toeplitz matrix, $V \in \F_q^{t \times n}$ is a Vandermonde matrix, and $D \in \F_q^{n \times n}$ is a diagonal matrix.

Since a Goppa code is a $\F_p$-subfield subcode, it is possible to express the syndrome function in terms of a parity-check matrix $\bar{H} \in \F_p^{mt \times n}$ using the so-called \emph{trace construction} (see e.g.\cite[Ch.~7, \S~7]{macwilliams-sloane}). This is useful to obtain a syndrome $\bar{s} \in \F_p^{mt}$ equivalent to $s \in \F_q^t$ above while keeping the arithmetic operations in $\F_p$ rather than $\F_q$, even though it is not immediately useful for decoding, at which point a syndrome over $\F_q$ has to be assembled by inverting the trace construction.

The \emph{syndrome decoding problem} consists of computing the error pattern $e$ given its syndrome $s_e$. Knowledge of the code structure in the form of the support $L$ and the polynomial $g$ makes this problem solvable in polynomial time, with some constraints relating the weight of $e$ to the degree of $g$.

\subsection{Polynomial Lattices}

Let $A \in \F_q[x]^{n \times m}$ be a polynomial matrix, and let $r$ denote its rank (i.e. assume that $A$ has $r$ linearly independent rows). The (polynomial) lattice $\lat(A)$ over $\F_q[x]$ spanned by the rows of $A$ is
\[
\lat(A) = \{ (u_0, \dots, u_{n-1}) A \in \F_q[x]^m \mid (u_0, \dots, u_{n-1}) \in \F_q[x]^n \}.
\]

The notion of length which we will use for $f \in \F_q[x]$ is $|f| = \deg(f)$. For polynomial vectors $v \in \F_q[x]^n$ we adopt the notion of maximal degree length: $|v| = \max_i |v_i|$. This notion is coarse enough that, contrary to integer lattices where finding even an approximation to the shortest vector by a constant factor is a hard problem~\cite{micciancio:focs}, reducing a basis for a polynomial lattice can be achieved in polynomial time. The following result by Mulders and Storjohann holds~\cite{mulders-storjohann}:

\begin{theorem}\label{thm:mulders-storjohann}
There exists an algorithm which finds the shortest nonzero vector in the $\F_q[x]$-module generated by the rows of $A$ with $O(m n r d^2)$ operations in $\F_q$, where $d = \max \{ \deg(A_{ij}) \mid 1 \leqslant i \leqslant n, \;  \leqslant j \leqslant m \}$.
\end{theorem}

The algorithm whose existence is established by Theorem~\ref{thm:mulders-storjohann} is based on converting a given lattice basis to the so-called weak Popov form, formally defined in Appendix~\ref{app:weak-Popov}, which also contains a description of Algorithm \ref{alg:weak-Popov} and its cost behavior in the context of decoding.

The weak Popov form is not the only way to find short vectors in a polynomial lattice, and in fact this is not critical to our proposal in this paper; for instance, the method by Lee and O'Sullivan \cite{lee-osullivan}, which is related to Gr\"{o}bner bases, would appear to be an alternative. Our choice of the Mulders-Storjohann method, which computes the weak Popov form, derives from its conceptual simplicity and ease of implementation, since the result turns out to be a natural generalization of Patterson's decoding algorithm described in the next section.

\section{Patterson's Decoding Method}\label{sec:patterson}

We briefly recapitulate Patterson's decoding algorithm~\cite{patterson}, which will provide the basis for the general algorithm we propose. The goal, of course, is to compute the error pattern $e$ given its syndrome $s_e$ and the structure of $\Gamma(L, g)$.

Let $q = 2^m$, and assume we are given a binary Goppa code $\Gamma(L, g)$ where the monic polynomial $g$ is irreducible. We define the Patterson locator polynomial $\sigma \in \F_q[x]$ as:
\begin{equation}\label{eq:patterson-sigma}
\sigma(x) := \prod_{e_i = 1}{(x - L_i)}.
\end{equation}
The name \emph{locator polynomial} comes from the fact that the roots of $\sigma$ clearly indicate where errors occurred, since $\sigma(L_j) = 0 \Leftrightarrow e_j = 1$. Taking the derivative of the formal power series underlying $\sigma$, we obtain
\begin{eqnarray*}
\sigma'(x) &=& \sum_{e_i = 1}{\prod_{\substack{e_j = 1\\j \neq i}}{(x - L_j)}}\\
           &=& \sum_{e_i = 1}{\dfrac{1}{x - L_i}\prod_{e_j = 1}{(x - L_j)}}\\
           &=& \sigma(x)\sum_i{\dfrac{e_i}{x - L_i}},
\end{eqnarray*}
and hence, in $\F_q[x]/g(x)$,
\begin{equation}\label{eq:key-binary}
\sigma'(x) = \sigma(x) s_e(x) \mod{g(x)}.
\end{equation}
This is called the \emph{key equation}, and now we discuss how to solve it.

Being a polynomial in characteristic~2, $\sigma(x)$ modulo $g(x)$ can be written as
\[
\sigma(x) = a_0(x)^2 + x a_1(x)^2
\]
for some $a_0(x)$, $a_1(x)$ with $\deg(a_0) \leqslant \lfloor t/2 \rfloor$ and $\deg(a_1) \leqslant \lfloor (t - 1)/2 \rfloor$, and hence
\[
\sigma'(x) = 2\,a_0(x)\,a_0'(x) + a_1(x)^2 + 2\,a_1(x)\,a_1'(x)\,x = a_1(x)^2,
\]
since the characteristic is~2.
Therefore
\begin{eqnarray*}
a_1(x)^2 &=& \sigma'(x) = \sigma(x) s_e(x)\\
         &=& \left(a_0(x)^2 + x a_1(x)^2\right) s_e(x) \mod g(x),
\end{eqnarray*}
whence
\begin{equation}\label{eq:bezout}
a_0(x) = a_1(x) v(x) \mod{g(x)}
\end{equation}
where $v(x)$ is a polynomial satisfying $v(x)^2 = x + 1/s_e(x) \mod g(x)$. Such a polynomial surely exists in characteristic~2 if $g(x)$ is square-free: if $g(x) = \prod_i{g_i(x)}$ where each $g_i(x)$ is irreducible, then $v(x) \bmod g_i(x)$ can be computed as a square root of $x + 1/s_e(x) \bmod g_i(x)$ viewed as an element of the finite field $\F_q[x]/g_i(x)$, and $v(x) \bmod g(x)$ can then be obtained by combining the results via the Chinese Remainder Theorem. We can thus assume that $\deg(v) < t = \deg(g)$.

The last equation is actually a B\'{e}zout relation $a_0(x) = a_1(x)v(x) + \lambda(x)g(x)$, which can be solved for $a_0(x)$ and $a_1(x)$ with the restriction $\deg(a_j) \leqslant \lfloor (t - j)/2 \rfloor$ (and also $\lambda(x)$ but it is not used) using the extended Euclidean algorithm. Solutions $(a_0,a_1)$ can also be seen as short vectors in the lattice spanned by the rows of the following matrix:
\[
A =
\left[
\begin{array}{cc}
g & 0\\
v & 1
\end{array}
\right]
\]
in the sense that the degrees of these polynomials are much smaller than uniformly random vectors, since $(\lambda, a_1)A = (\lambda g + a_1 v,a_1) = (a_0, a_1)$ for some $\lambda \in \F_q[x]$, by virtue of Equation~\ref{eq:bezout}. Therefore, Algorithm \ref{alg:weak-Popov} can be used to find candidate solutions $(a_0, a_1)$.

At first glance there is no guarantee that a short vector in the lattice generated by $A$ yields the desired solution; in other words, being short is a necessary condition, but in principle not a sufficient one.
However, the fact that in the binary case the minimum code distance is known~\cite{goppa} to be at least $2t + 1$ actually restricts $\sigma$ to a single candidate, so that Algorithm \ref{alg:weak-Popov} is bound to find it. Thus, decoding is always successful up to $t$ introduced errors.

\section{Decoding Codes over $\F_p$}\label{sec:extension}

We now show how to generalize Patterson's decoding algorithm so as to correct errors for codes defined over general prime fields. Thus, let $q = p^m$ for some prime $p$ and some $m > 0$, and assume we are given an irreducible Goppa code $\Gamma(L, g)$ over $\F_p$.

Let $\phi \in \F_p \setminus \{0\}$ be a constant scalar. We define the generalized error locator polynomial to be
\begin{equation}\label{eq:locator}
\sigma_\phi(x) := \prod_i{(x - L_i)^{e_i/\phi}}
\end{equation}
where the value $e_i/\phi$ is lifted from $\F_p$ to $\Z$ (i.e. the value $e_i/\phi \in \F_p$ that occurs as an exponent is taken to mean its corresponding least non-negative integer representative, which lies in range $0 \dots p-1$).
One can easily see that this definition actually coincides with Patterson error locator polynomials as defined by Equation~\ref{eq:patterson-sigma} for $p=2$. Lifting Equation~\ref{eq:locator} to the field of rational functions in characteristic 0 and taking the derivative, we have
\begin{eqnarray*}
\sigma_\phi'(x) &=& \sum_{j}{(e_j/\phi)(x - L_j)^{e_j/\phi - 1}\prod_{i \neq j}}{(x - L_i)^{e_i/\phi}}\\
           &=& (1/\phi)\sum_{j}{\dfrac{e_j}{x - L_j}\prod_i{(x - L_i)^{e_i/\phi}}}\\
           &=& (1/\phi)\sigma_\phi(x)\sum_{j}{\dfrac{e_j}{x - L_j}},
\end{eqnarray*}
which over $\F_q[x]$ reduces to
\begin{equation}\label{eq:modif-key}
\phi \sigma_\phi'(x) = \sigma_\phi(x) s_e(x) \mod{g(x)}.
\end{equation}
This is the $\phi$-th key equation of the proposed method, which generalizes Equation~\ref{eq:key-binary} to Goppa codes over $\F_p$.
The actual $\phi$ must be chosen so as to minimize the degree of $\sigma_\phi$ (and hence maximize the number of correctable errors). One cannot expect to know \emph{a priori} the value of $\phi$, but since there are only $p-1$ possibilities, the error correction strategy will be to try each of them in turn.

Notice that the maximum number of correctable errors can be, and usually is, less than $t$, since the degree of $\sigma_\phi$ exceeds the number of roots in the presence of multiple roots. The following theorem provides an upper bound for how many errors can be corrected by solving Equation~\ref{eq:modif-key} when the distribution of error magnitudes is not taken into account.

\begin{theorem}
The maximum number of errors that can be corrected by solving Equation~\ref{eq:modif-key} independently of the distribution of error magnitudes is $w = (2/p)t$.
\end{theorem}
\begin{proof}
Let $w_v$ denote the number of times the magnitude $v$ occurs in an error pattern of weight $w$, so that $\sum_v{w_v} = w$.
Since we are working with a Goppa code, the constraint for correctability is $\deg(\sigma_\phi) = \sum_v{(v/\phi)w_v} \leqslant t$. In the extreme situation when the weight of the error pattern reaches $w$, the most often error magnitude occurs $w_{\max} \geqslant w/(p-1)$ times, attaining the lower bound when all error magnitudes occur with equal frequency.
In that case, $\deg(\sigma_\phi) \leqslant \sum_v{(v/\phi)w_{\max}} = (1 + 2 + \cdots + (p-1)) w/(p-1) = wp/2 \leqslant t$, and hence no more than errors than $w = (2/p)t$ can be corrected independently of the distribution of magnitudes, as claimed.
\end{proof}

Since the proposed method coincides with Patterson's for $p = 2$, it is not surprising that $t$ errors can be corrected in that case. However, in characteristic~3 the number of potentially correctable errors is $(2/3)t$, non-deterministically exceeding the limit of $t/2$ errors attainable by previously known decoding methods for codes of degree $t$, except in the case of so called ``wild codes'' \cite{bernstein-lange-peters:wild} whereby the Goppa polynomial is a $(p-1)$-th power of an irreducible polynomial (our method, by contrast, applies when that polynomial is square-free, as we will see in Section~\ref{sec:solving}).

Despite the low general limit of $(2/p)t$ correctable errors for $p > 3$, it is still possible to exceed that limit in any odd characteristic if the distribution of error magnitudes is unbalanced. Indeed, all that is required to get a chance of uniquely decoding a word containing $w \leqslant t$ errors is that $\deg(\sigma_\phi) \leqslant t$ for some choice of $\phi$ and that the actual distance from the right codeword to any other codeword be at least $2w + 1$.

The actual number of correctable errors depends heavily on the distribution of error magnitudes and has to be computed in a case-by-case basis, always laying in the range $(2/p)t$ to $t$. In particular, if all error magnitudes are equal, in principle one could correct $t$ errors, even though this is a statistical rather than deterministic behavior.
This is especially useful for cryptographic applications involving an all-or-nothing transform~\cite{rivest}, as it happens e.g. for a semantically secure encryption scheme involving the McEliece one-way trapdoor function~\cite{mceliece} and the Fujisaki-Okamoto conversion~\cite{fujisaki-okamoto}. In such scenarios, the magnitudes of the introduced errors can be chosen to be all or nearly all equal by convention, making the proposed decoder attractive for its higher decodability bound, under the explicit assumption that decoding them remains hard.

\subsection{Solving the Key Equation}\label{sec:solving}

We now focus on actually solving Equation~\ref{eq:modif-key}.
Being a polynomial in characteristic~$p$, $\sigma_\phi(x)$ can be written as
\begin{equation}\label{eq:decompose}
\sigma_\phi(x) = \sum_{k=0}^{p-1}{x^k a_k(x)^p}
\end{equation}
for some $a_k(x)$ with $\deg(a_k) \leqslant \lfloor (t - k)/p \rfloor$, $0 \leqslant k \leqslant p-1$, and hence
\begin{eqnarray*}
\sigma_\phi'(x) &=& \sum_{k=0}^{p-1}{\left(k x^{k-1} a_k(x)^p + p x^k a_k(x)^{p-1} a_k'(x)\right)}\\
           &=& \sum_{k=1}^{p-1}{k x^{k-1} a_k(x)^p}
\end{eqnarray*}
since the characteristic is $p$.
Therefore
\begin{eqnarray*}
\phi\sum_{k=1}^{p-1}{k x^{k-1} a_k(x)^p} &=& \phi\sigma_\phi'(x) \;=\; \sigma_\phi(x) s_e(x)\\
&=& \left(\sum_{k=0}^{p-1}{x^k a_k(x)^p}\right) s_e(x)
\bmod{g(x)},
\end{eqnarray*}
whence
\begin{equation}\label{eq:lattice}
a_0 + \sum_{k=1}^{p-1}{a_k(x) v_k(x)} = 0 \mod g(x)
\end{equation}
where the $v_k(x)$ are polynomials satisfying $v_k(x)^p = x^k - \phi k x^{k-1} / s_e(x) \mod g(x)$. Such polynomials surely exist in characteristic~$p$ if $g(x)$ is square-free: if $g(x) = \prod_i{g_i(x)}$ where each $g_i(x)$ is irreducible, then $v_k(x) \bmod g_i(x)$ can be computed as a $p$-th root of $x^k - \phi k x^{k-1} / s_e(x) \bmod g_i(x)$ viewed as an element of the finite field $\F_q[x]/g_i(x)$, and $v_k(x) \bmod g(x)$ can then be obtained by combining the results via the Chinese Remainder Theorem. We can thus assume that $\deg(v_k) < t = \deg(g)$.

The Diophantine equation~\ref{eq:lattice} has to be solved for $a_k(x)$ with the stated restriction on their degrees. Solutions $(a_0, a_1, \dots, a_{p-1})$ can be seen as short vectors in the lattice spanned by the rows of the matrix
\begin{equation}\label{eq:lattice-basis}
A_\phi =
\left[
\begin{array}{ccccc}
g        & 0      & 0      & \dots  & 0\\
-v_1     & 1      & 0      & \dots  & 0\\
-v_2     & 0      & 1      & \dots  & 0\\
\vdots   & \vdots & \vdots & \ddots & \vdots\\
-v_{p-1} & 0      & 0      & \dots  & 1
\end{array}
\right],
\end{equation}
since, by virtue of Equation~\ref{eq:lattice}, one has $(\lambda, a_1, \dots, a_{p-1}) A_\phi = (\lambda g - \sum_{k=1}^{p-1} a_k(x) v_k(x), a_1, \dots, a_{p-1}) = (a_0, a_1, \dots, a_{p-1})$ for some $\lambda \in \F_q[x]$.
Therefore, Algorithm \ref{alg:weak-Popov} can be used to find candidate solutions $(a_0, \ldots, a_{p-1})$.

The method is applicable whenever one can actually invert $s \bmod g$ and then compute the $p$-roots needed to define the $v_k$ polynomials. This is always the case when $g$ is irreducible, but not exclusively so. Indeed, to compute the $v_k$ it suffices that $g$ is square-free and that $s$ is invertible modulo each of the irreducible factors of $g$, since in this case the $v_k$ can be easily computed modulo those irreducible factors and finally recovered via the Chinese Remainder Theorem. 

Theorem~\ref{thm:mulders-storjohann} ensures a cost not exceeding $O(p^3 t^2)$ $\F_q$ operations for computing short vectors in $\lat(A_\phi)$.

\subsection{Estimating the Success Probability}\label{sec:success-prob}

Regrettably the ability to find shorts vectors in lattice $\lat(A_\phi)$ does not mean that any such vector yields a solution to Equation~\ref{eq:modif-key}. We will now see that, fortunately, the proposed method has a surprisingly favourable probability of finding the right $(a_0,\ldots,a_{p-1})$ that solves Equation~\ref{eq:modif-key}.

As a cautionary note, we stress that a full theoretical analysis of the failure probability remains, to the best of our knowledge, an open problem. Because the distance notion here is not Euclidean, but rather that of Hamming, Minkowski's theorem on the existence of a lattice point in any large enough convex set does not appear to apply in our case. Also, the analysis would also appear to require a detailed theory on the distribution of the average distance between a vector and the closest codeword whose magnitudes satisfy some constraint (like being all equal, or following a highly skewed distribution), which to the best of our knowledge is too an open problem.
As a consequence, the failure probability estimates we provide are conjectured on an empirical basis. Namely, they result from experiments conducted on a large number of random Goppa codes at which our decoding method is targeted, and for each of those codes, on a large number of decoding attempts on random error patterns following the particular magnitude constraint (equal for all error positions) for which the decoder works best.

In a successful decoding, the reduced basis for lattice $\lat(A_\phi)$ leads to candidates for $\sigma_\phi$ with degree $t$ onward, of which of course only the candidate with the smallest degree is the correct one. Spurious candidates of degree close to $t$ result from random-looking short (albeit not shortest) vectors in the reduced basis and are usually harmless. But the fact that those short vectors are ``random-looking'' means they are also a threat: if by chance they are such that the coefficient of the highest-degree term in the associated spurious $\sigma_\phi$ vanishes, $\deg(\sigma_\phi)$ becomes $t$ or less. Since this is connected with the vanishing of a coefficient from $\F_q$, this event happens with probability $1/q$ assuming that short spurious vectors in $\lat(A_\phi)$ are approximately uniformly distributed.

In general, when trying to correct $w < t$ errors of equal magnitude for a uniformly random irreducible Goppa code, the top $t + 1 - w$ coefficients in the spurious $\sigma_\phi$ must vanish to interfere with the decoding process, whence the probability of successful decoding is roughly $1 - 1/q^{t + 1 - w}$. This matches the empirically observed behaviour of the proposed method in odd characteristic. Not surprisingly, the method always succeeds for binary codes, since it reduces to Patterson's algorithm and the minimum code distance is known to be at least $2t + 1$.

Table~\ref{tab:results} illustrates the results of experiments in Magma~\cite{magma} supporting the conjecture that the probability of successful decoding is roughly $\Pr_{suc} := 1 - 1/q^{t + 1 - w}$.
For each quadruple $(p, m, t, w)$, a set of 10000 Goppa codes of maximum length $n = p^m$ and degree $t$ plus an error pattern of length $n$, weight $w$ and all magnitudes equal (to a single random value in $\F_p \setminus \{0\}$) were randomly generated, and the proposed method was then applied to decode the syndrome of that error pattern. The predicted number of successful decodings, $N_{pre}$, is then compared with the actually observed number $N_{obs}$ of successful decodings.
For each combination of $p$ and $m$, the first listed $t$ is the largest integer satisfying $p^m - mt > 0$. The number $w$ of introduced errors in then decreased starting from $t$ until the probability of success exceeds 0.9999.
Since $\Pr_{suc}$ is close to 1 for large $q$, reasonably small values are chosen for all parameters so that the probability of decoding failure is large enough to be easily discernible. This is also the reason why more detail is provided for smaller parameters. We omit the results for characteristic~2, since in all tests we conducted no decoding failure was observed, as expected. We stress that these examples are not meant by any means for practical cryptographic use.

\begin{table}[htp]\centering
\caption{Experimental assessment of the probability of decoding success}\label{tab:results}
\begin{tabular}{ccccccc}\hline
$p$ & $m$ & $t$ & $w$ & $\Pr_{suc}$ & $N_{pre}$ & $N_{obs}$\\
\hline
 3  &  3  &   8 &   8 & 0.962963    &      9630 &  9670\\
 3  &  3  &   8 &   7 & 0.998628    &      9986 &  9992\\
 3  &  3  &   8 &   6 & 0.999949    &      9999 &  9999\\
\\
 3  &  3  &   7 &   7 & 0.962963    &      9630 &  9639\\
 3  &  3  &   7 &   6 & 0.998628    &      9986 &  9989\\
 3  &  3  &   7 &   5 & 0.999949    &      9999 & 10000\\
\\
 3  &  3  &   6 &   6 & 0.962963    &      9630 &  9645\\
 3  &  3  &   6 &   5 & 0.998628    &      9986 &  9991\\
 3  &  3  &   6 &   4 & 0.999949    &      9999 & 10000\\
\\
 3  &  4  &  20 &  20 & 0.987654    &      9877 &  9883\\
 3  &  4  &  20 &  19 & 0.999848    &      9998 &  9997\\
 3  &  4  &  20 &  18 & 0.999998    &     10000 & 10000\\
\\
 5  &  2  &  12 &  12 & 0.960000    &      9600 &  9612\\
 5  &  2  &  12 &  11 & 0.998400    &      9984 &  9985\\
 5  &  2  &  12 &  10 & 0.999936    &      9999 & 10000\\
\\
 5  &  3  &  41 &  41 & 0.992000    &      9920 &  9924\\
 5  &  3  &  41 &  40 & 0.999936    &      9999 & 10000\\
\\
 7  &  2  &  24 &  24 & 0.997085    &      9971 &  9989\\
 7  &  2  &  24 &  23 & 0.999992    &      9999 & 10000\\
\\
11  &  2  &  60 &  60 & 0.991736    &      9917 &  9922\\
11  &  2  &  60 &  59 & 0.999932    &      9999 &  9999\\
\hline
\end{tabular}
\end{table}

Decoding $w \leqslant (2/p)t$ errors of uniformly random magnitude for a uniformly random irreducible code is of course always successful for $p > 3$, since in that case $w$ does not exceed half the minimum code distance, which is at least $t + 1 > (4/p)t + 1 \geqslant 2w + 1$.

\subsection{Computing the Error Magnitudes}\label{sec:error-eval}

In contrast to generic alternant decoding methods, there is no need to compute an error evaluator polynomial to obtain the error magnitudes in the current proposal. After obtaining $\sigma_\phi(x)$ and finding its roots $L_j$, all that is needed to compute the corresponding error values $e_j$ is to determine the multiplicity $\mu_j$ of each root, since one can see from Equation~\ref{eq:locator} that $e_j = \phi\mu_j$.

Computing $\mu_j$ is accomplished by determining how many times $(x - L_j) \mid \sigma_\phi(x)$ whenever $\sigma_\phi(L_j) = 0$, or alternatively by finding the highest derivative of $\sigma_\phi$ such that $\sigma_\phi^{(h)}(L_j) = 0$ (and setting $\mu_j \gets h$).

Since the value of $\phi$ is not known \emph{a priori}, and even in scenarios where it is actually known beforehand, an additional syndrome check is necessary for each guessed $\phi$, and the process usually must check all possible values of $\phi$ anyway since more than one solution may exist.

\subsection{The Completed Decoder}\label{sec:decoder}

We are finally ready to state the full decoding method explicitly in Algorithm~\ref{alg:decoder}.
It can be seen as a list decoding algorithm with possible failures.
The polynomial decomposition of Equation~\ref{eq:decompose} immediately suggests a simple and efficient way to compute the $p$-th roots needed at Step~\ref{step:p-th-root}, namely, precompute $r(x) \gets \sqrt[p]{x} \mod g(x)$ and $r(x)^k \bmod g(x)$, and then compute the $p$-th root of $z(x) := \sum_k{x^k z_k(x)^p}$ as $\sqrt[p]{z(x)} \bmod g(x) \gets \sum_k{r(x)^k z_k(x)}$. The test in Step~\ref{step:test-inverse} is unnecessary if $g$ is irreducible.
To find the zeroes of $\sigma_\phi$ in Step~\ref{step:find-zeroes} one can use the Chien search technique~\cite{chien}, in which case the multiplicities of each root can be determined as part of the search, or the Berlekamp trace algorithm~\cite{berlekamp:zeroes}.

\begin{algorithm}
\caption{Decoding $p$-ary square-free Goppa codes}\label{alg:decoder}
\begin{algorithmic}[1]
\Input $\Gamma(L, g)$, a Goppa code over $\F_p$ where $g$ is square-free.
\Input $H \in \F_q^{r \times n}$, a parity-check matrix in the form of Equation~\ref{eq:parity-check}.
\Input $c' = c + e \in \F_p^n$, the received codeword with errors.
\Output a list of corrected codewords $c \in \Gamma(L, g)$ ($\varnothing$ upon failure).
\State $s^{\tp} \gets H c'^{\tp} \in \F_q^n$, $s_e(x) \gets \sum_i{s_i x^i}$. \ZComment N.B. $H c'^{\tp} = H e^{\tp}$.
\If{$\nexists \; s_e^{-1}(x) \bmod g(x)$}\label{step:test-inverse}
    \State \Return $\varnothing$ 
\EndIf
\State $S \gets \varnothing$
\For{$\phi \gets 1 \; \textbf{to} \; p-1$}\label{step:guess-scale} \ZComment guess the correct scale factor $\phi$
    \For{$k \gets 1 \; \textbf{to} \; p-1$}
        \State $u_k(x) \gets x^k + \phi k x^{k-1} / s_e(x) \mod {g(x)}$
		\State\label{step:p-th-root} $v_k(x) \gets \sqrt[p]{u_k(x)} \bmod g(x)$
	\EndFor
	\State Build the lattice basis $A_\phi$ defined by Equation~\ref{eq:lattice-basis}.
	\State Apply Algorithm~\ref{alg:weak-Popov} to reduce the basis of $\lat(A_\phi)$.
	\For{$i \gets 1 \; \textbf{to} \; p$}\label{step:take-shortest}
	    \State Let $a$ denote the $i$-th row of the reduced basis of $\lat(A_\phi)$ 
	    \For{$j \gets 0 \; \textbf{to} \; p-1$}
		    \If{$\deg(a_j) > \lfloor (t - j)/p \rfloor$}
	            \State \textbf{try next} $i$ \ZComment not a solution
		    \EndIf
		\EndFor
		\State $\sigma_\phi(x) \gets \sum_j{x^j a_j(x)^p}$
		\State\label{step:find-zeroes} Compute the set $J$ such that $\sigma_\phi(L_j) = 0$, $\forall j \in J$.
		\For{$j \in J$}
		    \State Compute the multiplicity $\mu_j$ of $L_j$.
		    \State $e_j \gets \phi \mu_j$
		\EndFor
		\If{$H e^{\tp} = s^{\tp}$}
			\State $S \gets S \cup \{c' - e\}$
		\EndIf
	\EndFor
    \State \Return $S$
\EndFor
\end{algorithmic}
\end{algorithm}

\section{Conclusion}\label{sec:conclusion}

We described a new decoding algorithm for square-free (in particular, irreducible) Goppa codes of degree $t$ over $\F_p$ that can correct $(2/p)t$ errors in general, and up to $t$ errors for certain distributions of error magnitudes of cryptographic interest. By attaining an correction capability of $(2/3)t$ errors in characteristic~3 with high probability, our method outperforms the best previously known decoder for that case, and suggests that the corresponding average distance to the closest codeword is at least $(4/3)t + 1$ for most irreducible ternary Goppa codes. Regardless of the characteristic, our proposal can correct a still larger number of errors that approaches (and probabilistically reaches) $t$ as the distribution of error magnitudes becomes ever more skewed toward the predominance of some individual value. The method can be viewed as generalizing Patterson's binary decoding procedure, and is similarly efficient in practice.

A further increase in the number of correctable errors may be possible by resorting to list decoding and by extracting more information from the decoding process along the lines proposed by Bernstein~\cite{bernstein}.
This in principle might enable the correction of approximately $n - \sqrt{n(n - (4/p)t)}$ errors in general, and possibly as many as $n - \sqrt{n(n - 2t - 2)}$ errors depending on the distribution of error magnitudes.
Furthermore, the ability to correct close to $t$ errors with high probability means that smaller keys might be adopted for coding-based cryptosystems. Properly chosen parameters would keep the probability of decoding failure below the probability of breaking the resulting schemes by random guessing, while maintaining the security at the desired level.
We leave the investigation of such possibilities for future research.

\bibliographystyle{IEEEtran}
\bibliography{GoppaDec}

\appendices

\section{The Weak Popov Form}\label{app:weak-Popov}

For ease of reference, we provide here a concise description of the Mulders-Storjohann polynomial lattice reduction algorithm based on the weak Popov form. We closely follow the exposition in~\cite{mulders-storjohann}, while attempting to make our description more implementation-friendly.

\begin{definition} For $1 \leqslant i \leqslant n$ the \emph{$i$-th pivot index vector $I^M$} of a matrix $M \in \F_q[x]^{n \times m}$ is defined as follows: if $M_{ij} = 0$ for all $1 \leqslant j \leqslant m$, then $I_i^M = 0$, otherwise
\begin{enumerate}
\item $\deg(M_{ij}) \leqslant \deg(M_{i,I_i^M})$ for $1 \leqslant j < I_i^M$,
\item $\deg(M_{ij}) < \deg(M_{i,I_i^M})$ for $I_i^M < j \leqslant m$.
\end{enumerate}
\end{definition}

\begin{definition} The \emph{carrier set} $C^M$ of a matrix $M \in \F_q[x]^{n \times m}$ is the set $\{1 \leqslant i \leqslant n \mid I_i^M \neq 0\}$.
%
The \emph{$i$-th pivot element} of $M$, denoted $P_i^M$, is the element $P_i^M := M_{i,I_i^M}$ when $I_i^M \neq 0$, otherwise $P_i^M := 0$.
\end{definition}

\begin{definition}
A matrix $M \in \F_q[x]^{n \times m}$ is said to be in \emph{weak Popov form} if the positive pivot indices of $M$ are all
different, i.e. if $\forall k, \ell \in C^M: k \neq \ell \Rightarrow I_k^M \neq I_\ell^M$.
\end{definition}

The following theorem establishes that writing a matrix in weak Popov form yields short vectors in the lattice spanned by its rows.
\begin{theorem}[\cite{mulders-storjohann}]
If matrix $M \in \F_q[x]^{n \times m}$ is in weak Popov form and $\ell$ is such that $\deg(P_\ell^M) = \min_{1 \leqslant i \leqslant n}\{\deg(P_i^M)\}$, then all vectors in the $\F_q[x]$-module generated by the rows of $M$ have degree at least $\deg(P_\ell^M)$.
\end{theorem}
\begin{proof}
See~\cite[Lemma~8.1]{mulders-storjohann}.
\end{proof}

If $k \in C^M$, $\ell \neq k$ and $\deg(M_{\ell,I_k^M}) \geqslant \deg(P_k^M)$, there are unique $c \in \F_q$ and $e \in \N$ such that $\deg(M_{\ell,I_k^M} - c x^e P_k^M) < \deg(M_{\ell,I_k^M})$.
In that case we call the operation $M_{\ell} \gets M_{\ell} - c x^e M_{k}$ the \emph{simple transformation} of row $k$ on row $\ell$. If $I_\ell^M = I_k^M$, the transformation is called of the \emph{first kind}. 
Then an efficient algorithm to put a matrix in weak Popov form stems from the following observation:
\begin{theorem}[\cite{mulders-storjohann}]
$M \in \F_q[x]^{n \times m}$ is not in weak Popov form iff one can apply a simple transformation of the first kind on $M$, that is, not all non-zero pivot indices of $M$ are different.
\end{theorem}
\begin{proof}
See~\cite[Lemma~2.1]{mulders-storjohann}.
\end{proof}

Therefore, all one has to do to obtain the weak Popov form of a matrix $M$ is to repeatedly check if $M$ is already in the weak Popov form (by testing if all nonzero pivot indices are different) and, if it is not, apply a simple transformation of the first kind on it.

This process is summarised in Algorithm~\ref{alg:weak-Popov} for a matrix in the form of Equation~\ref{eq:lattice-basis}, where $n = m = p$, the expected rank is $r \leqslant p$, and the degree of all rows is bounded by $d \leqslant t$. By Theorem~\ref{thm:mulders-storjohann}, its complexity is $O(p^3 t^2)$ $\F_q$ operations at most.
Here $\lead(P)$ denotes the leading coefficient of $P \in \F_q[x]$ and $\rep(I^A)$ denotes the number of occurrences of the most frequent value among the nonzero components of $I^A$, i.e. $\rep(I^A) := \max\{\#\{j \mid I_j^A = v\} \mid v \neq 0\}$.

Written in this form, Algorithm~\ref{alg:weak-Popov} is strikingly similar to the modified Euclidean algorithm usually employed in the decoding of alternant codes~\cite{sugiyama-kasahara-hirasawa-namekawa:decoding}, and actually coincides with that method for $p = 2$.

\begin{algorithm}
\caption{Computing the weak Popov form}\label{alg:weak-Popov}
\begin{algorithmic}[1]
\Input $A \in \F_q[x]^{p \times p}$ in the form of Equation~\ref{eq:lattice-basis}.
\Output weak Popov form of $A$.
\State \ZComment Compute $I^A$:
\For{$j \gets 1 \; \textbf{to} \; p$}
    \State $I_j^A \gets \; \textbf{if} \; \deg(A_{j, 1}) > 0 \; \textbf{then} \; 1 \; \textbf{else} \; j$
\EndFor
\State \ZComment Put $A$ in weak Popov form:
\While{$\rep(I^A) > 1$} \label{stat:loop-r} 
    \For{$k \gets 1 \To p \SuchThat I_k^A \neq 0$} \label{stat:loop-alpha-k}
	    \For{$\ell \gets 1 \To p \SuchThat \ell \neq k$} \label{stat:loop-alpha-l}
	        \While{$\deg(A_{\ell, I_k^A}) \geqslant \deg(A_{k, I_k^A})$} \label{stat:loop-lambda}
                \State $c \gets \lead(A_{\ell, I_k^A})/\lead(A_{k, I_k^A})$
                \State $e \gets \deg(A_{\ell, I_k^A}) - \deg(A_{k, I_k^A})$ 
                \State $A_\ell \gets A_\ell - c x^e A_k$
	        \EndWhile
            \State \ZComment Update $I_\ell^A$ and hence $\rep(I^A)$ if necessary:
            \State $d \gets \max\{\deg(A_{\ell, j}) \mid j = 1, \dots, p\}$
            \State $I_\ell^A \gets \max\{j \mid \deg(A_{\ell, j}) = d\}$
	    \EndFor
    \EndFor
\EndWhile
\State \Return $A$
\end{algorithmic}
\end{algorithm}


\section{Decoding Other Families of Codes?}\label{sec:bch}

For completeness, we briefly discuss whether and how one might attempt to use similar methods to decode a different family of alternant codes, including BCH codes and their permuted and/or punctured versions.

Let $L \in \F_q^n$ be a sequence of $n \leqslant q$ distinct nonzero elements, let $D \in \F_q^n$ be a sequence of nonzero elements, and let $H = \vdm(L) \diag(D)$. For any word $e \in \F_p^n$ we define the corresponding \emph{alternant $r$-syndrome} polynomial $s_e \in \F_q[x]$ to be $s_e(x) := \sum_{i=0}^{r-1}{s_i x^i}$ where $s^{\tp} := H e^{\tp}$, i.e.
\[
s_i = \sum_{j=0}^{n-1}{e_j D_j L_j^i}.
\]
The alternant code $\mathcal{A}(L, D, r)$ consists of the set $\{ e \in \F_p^n \mid s_e(x) \equiv 0\}$.

Using the formula for the sum of a geometric sequence $\sum_{i=0}^{r-1}{u^i} = (1 - u^r)/ (1 - u)$ whereby $\sum_{i=0}^{r-1}{L_j^i x^i} = (1 - x^r L_j^r)/ \; \; \; (1 - x L_j) \equiv 1/(1 - x L_j) \bmod x^r$, one can see that
\begin{eqnarray*}
s_e(x) &=& \sum_{i=0}^{r-1}\sum_{j=0}^{n-1}{e_j D_j L_j^i x^i} = \sum_{j=0}^{n-1}{e_j D_j \sum_{i=0}^{r-1}{L_j^i x^i}}\\
       &\equiv& \sum_{j=0}^{n-1}{\dfrac{e_j D_j}{1 - x L_j}} \mod{x^r}.
\end{eqnarray*}

The subfamily we will be interested in is that of alternant codes satisfying the restriction $\xi_j := D_j/L_j \in \F_p \setminus \{0\}$ for all $j$, so that each value $\xi_j$ can be lifted to $\Z$ with a representative in range $1 \dots p-1$.

Let $\phi \in \F_p \setminus \{0\}$ be a constant scalar. We define the generalized error locator polynomial for this family as
\begin{equation}\label{eq:pplocator}
\sigma_\phi(x) := \prod_i{(1 - x L_i)^{e_i\xi_i/\phi}}.
\end{equation}
The error positions are revealed by the \emph{inverses} of the components of $L$, which are the roots of this polynomial. This definition coincides with the usual alternant error locator polynomial when $p = 2$, in which case $D = L$ (hence, a permuted and/or punctured subcode of a binary BCH code).

Taking the derivative of the formal power series underlying $\sigma_\phi$ in Equation~\ref{eq:pplocator}, we get
\begin{eqnarray*}
\sigma_\phi'(x) &=& \sum_{j}{(e_j\xi_j/\phi)(1 - xL_j)^{e_j\xi_j/\phi - 1}(-L_j)}\\
           & & \prod_{i \neq j}{(1 - xL_i)^{e_i\xi_i/\phi}}\\
           &=& -(1/\phi)\sum_{j}{\dfrac{e_j D_j}{1 - xL_j}\prod_i{(1 - xL_i)^{e_i\xi_i/\phi}}}\\
           &=& -(1/\phi)\sigma_\phi(x)\sum_{j}{\dfrac{e_j D_j}{1 - xL_j}},
\end{eqnarray*}
which over $\F_q[x]$ reduces to
\begin{equation}\label{eq:modif-putz-putz}
-\phi \sigma_\phi'(x) = \sigma_\phi(x) s_e(x) \mod{x^r}.
\end{equation}
This is the $\phi$-th key equation for this family of codes.

Now most of the techniques developed above for Goppa codes can be applied to solve Equation~\ref{eq:modif-putz-putz}. The main difference is that the error magnitudes are computed as a function of the multiplicity $\mu_j$ of a root $1/L_j$ of $\sigma_\phi$ as $e_j = \phi \mu_j/\xi_j$.

Writing $\sigma_\phi(x) = \sum_{k=0}^{p-1}{x^k a_k(x)^p}$ for some $a_k(x)$ with $\deg(a_k) \leqslant \lfloor (r - k)/p \rfloor$, solutions to Equation~\ref{eq:modif-putz-putz} can be found as short vectors $(a_0, a_1, \dots, a_{p-1})$ in the polynomial lattice spanned by the rows of the matrix
\[
A_\phi =
\left[
\begin{array}{cccc}
x^r      & 0      & \dots  & 0\\
-v_1     & 1      & \dots  & 0\\
\vdots   & \vdots & \ddots & \vdots\\
-v_{p-1} & 0      & \dots  & 1
\end{array}
\right]
\]
where the $v_k(x)$ are polynomials satisfying $v_k(x)^p = x^k - \phi k x^{k-1} / s_e(x) \mod{x^r}$, provided that these exist.

Here the major obstacle for this technique becomes apparent: inverting $s_e(x) \bmod x^r$ is usually fine, but computing the $v_k(x)$ polynomials is only very seldom possible. Specifically, assuming that $x^k - \phi k x^{k-1} / s_e(x) \mod{x^r}$ are uniformly distributed polynomials in $\F_q[x]/x^r$ for a random code of this family, the probability that it is a $p$-th power mod $x^r$ is only about $(q^{r/p}/q^r)^{p-1} = p^{-mr(p-1)^2/p}$, corresponding to the vanishing of all but a fraction $1/p$ of the $r$ coefficients of each of the $p-1$ polynomials needed to build matrix $A$.

Therefore there is scant chance that this would work in practice, except possibly for some highly contrived code whose syndromes lead to suitable radicands with high probability. It is an open problem whether such codes exist and, if so, what they might look like.


\end{document}